\documentclass{cccg20}
\usepackage{amssymb,amsmath}
\usepackage{microtype}
\usepackage{wrapfig}
\usepackage{graphicx}
\usepackage{xspace}
\usepackage{multicol}
\usepackage{cite}
\usepackage[table]{xcolor}
\usepackage[ruled,vlined,linesnumbered]{algorithm2e}
\graphicspath{{img/}}
\usepackage{subcaption}
\usepackage{enumitem}
\usepackage{hyperref}

\newcommand{\RE}{\mathbb{R}}

\newcommand{\bigOh}{\mathcal{O}}

\newcommand{\OO}[1]{O\kern-2pt\left(#1\right)}  

\newcommand{\setP}{\ensuremath{P}\xspace}
\newcommand{\setR}{\ensuremath{R}\xspace}

\newcommand{\ie}{\emph{i.e.},\xspace}
\newcommand{\eg}{\emph{e.g.},\xspace}

\newcommand{\wwlog}{\emph{w.l.o.g.}\xspace}

\newcommand{\numNE}{\kappa}

\newcommand{\spread}{\Delta}

\newcommand{\metricSet}{\mathcal{X}}
\newcommand{\metricFunc}{\textup{\textsf{d}}}
\newcommand{\metricSpace}{$(\metricSet,\metricFunc)$\xspace}
\newcommand{\ddim}{\textup{ddim}(\metricSet)}

\newcommand{\dist}[2]{\metricFunc(#1,#2)}
\newcommand{\nn}[1]{\textup{nn}(#1)}
\newcommand{\nenemy}[1]{\textup{ne}(#1)}

\newcommand{\CNN}{\textup{CNN}\xspace}
\newcommand{\FCNN}{\textup{FCNN}\xspace}
\newcommand{\SFCNN}{\textup{SFCNN}\xspace}
\newcommand{\NET}{\textup{NET}\xspace}

\newcommand{\MSS}{\textup{MSS}\xspace}
\newcommand{\VSS}{\textup{VSS}\xspace}
\newcommand{\RSS}{\textup{RSS}\xspace}

\definecolor{yellowcd}{RGB}{252, 229, 30}
\definecolor{bluecd}{RGB}{51, 0, 68}

\title{Social Distancing is Good for Points too!}

\author{Alejandro Flores-Velazco\thanks{University of Maryland, College Park, {\tt afloresv@cs.umd.edu}}}

\index{Flores-Velazco, Alejandro}

\begin{document}
\thispagestyle{empty}
\maketitle

\begin{abstract}
The \emph{nearest-neighbor rule} is a well-known classification technique that,
given a training set \setP of labeled points, classifies any unlabeled query point with the label of its closest point in \setP. The \emph{nearest-neighbor condensation} problem aims to reduce the training set without harming the accuracy of the nearest-neighbor rule.

\FCNN is the most popular algorithm for condensation. It is heuristic in nature, and theoretical results for it are scarce. In this paper, we settle the question of whether reasonable upper-bounds can be proven for the size of the subset selected by \FCNN. First, we show that the algorithm can behave poorly when points are too close to each other, forcing it to select many more points than necessary. We then successfully modify the algorithm to avoid such cases, thus imposing that selected points should ``keep some distance''. This modification is suffi\-cient to prove useful upper-bounds, along with approximation guarantees for the algorithm.
\end{abstract}

%
%
%
\section{Introduction}

In the context of non-parametric classification, a \emph{training set} \setP consists of $n$ points in a metric space \metricSpace, with domain $\metricSet$ and distance function $\metricFunc: \metricSet^2 \rightarrow \RE^{+}$. Additionally, \setP is partitioned into a finite set of \emph{classes} by associating each point $p \in \setP$ with a \emph{label} $l(p)$, indicating the class to which it belongs. Given an \emph{unlabeled} query point $q \in \metricSet$, the goal of a \emph{classifier} is to predict $q$'s label using the training set \setP.

The \emph{nearest-neighbor rule} is among the best-known classification techniques~\cite{fix_51_discriminatory}. It assigns a query point the label of its closest point in $\setP$, according to the metric~$\metricFunc$. The nearest-neighbor rule exhibits good classification accuracy both experimentally and theoretically \cite{stone1977consistent,Cover:2006:NNP:2263261.2267456,devroye1981inequality}, but it is often criticized due to its high space and time complexities. Clearly, the training set \setP must be stored to answer nearest-neighbor queries, and the time required for such queries depends to a large degree on the size and dimensionality of the data. These drawbacks inspire the question of whether it is possible replace \setP with a significantly smaller~subset, without significantly reducing the classification accuracy under the nearest-neighbor rule. This problem has been widely \mbox{studied, and it} is often called \emph{nearest-neighbor~condensation}~\cite{Hart:2006:CNN:2263267.2267647,ritter1975algorithm,gottlieb2014near,DBLP:conf/jcdcg/Toussaint02}.

\paragraph{Related work.}
A subset $\setR \subseteq \setP$ is said to be \emph{consistent} if and only if for every $p \in \setP$~its nearest-neighbor in \setR is of the same class as $p$. Intuitively, \setR is consistent~\cite{Hart:2006:CNN:2263267.2267647} if and only if all points of \setP are correctly classified using the nearest-neighbor rule over \setR. Formally, the problem of nearest-neighbor condensation consists of finding an ideally small consistent subset of \setP.

It is known that the problem of computing consistent subsets of minimum cardinality is NP-hard~\cite{Wilfong:1991:NNP:109648.109673,Zukhba:2010:NPP:1921730.1921735,khodamoradi2018consistent}. However, there exists an algorithm called \NET~\cite{gottlieb2014near} that computes a tight approximation of the minimum cardi\-nality consistent subset. Yet, this algorithm is not practical, and it is often outperformed on real-world training sets |with respect to both their runtime and size of the selected subsets| by simple heuristics for condensation.

\begin{figure*}[h!]
    \centering
    \begin{subfigure}[b]{.25\linewidth}
        \centering\includegraphics[width=.9\textwidth]{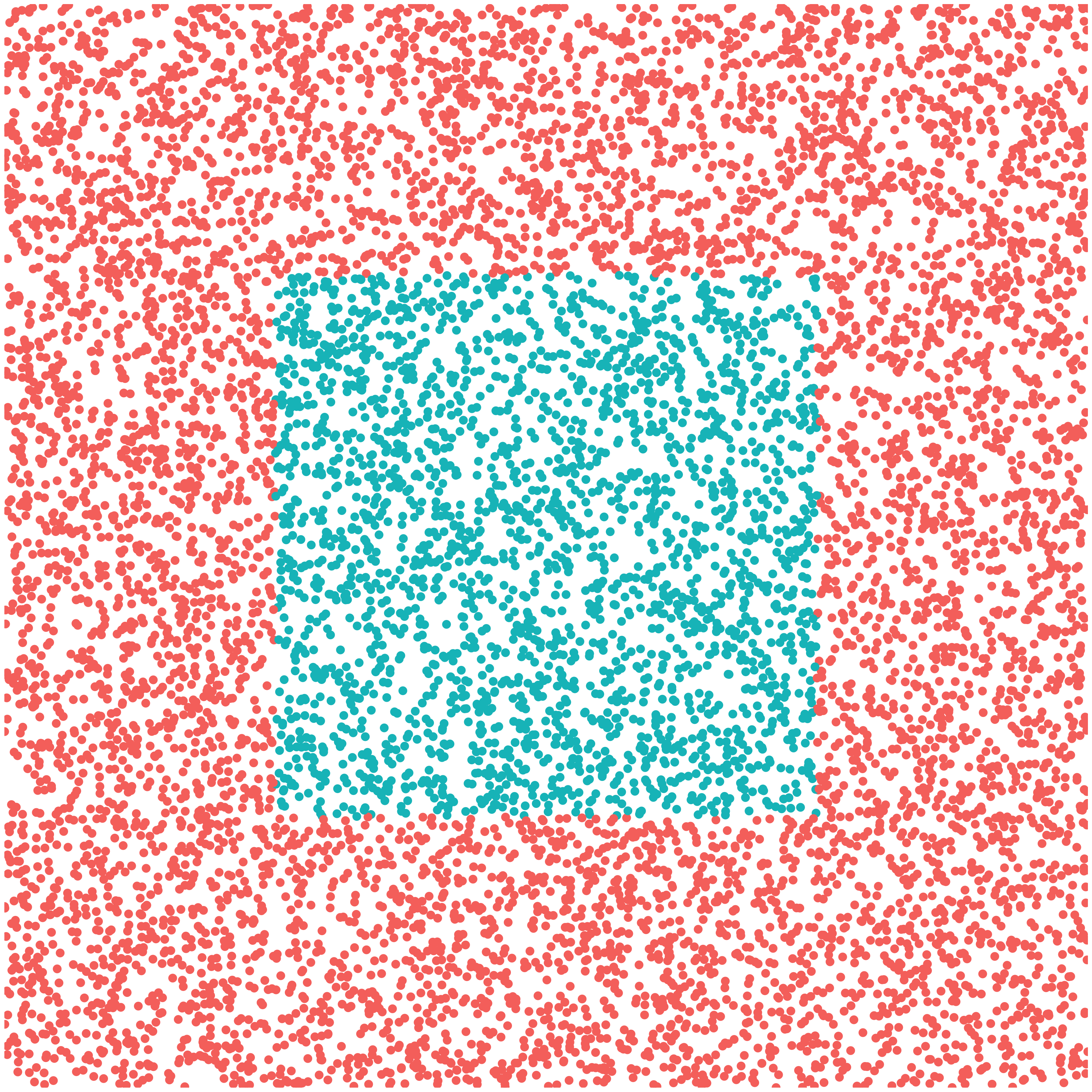}
        \caption{Training set ($10^4$\,pts)}\label{fig:algexample:dataset}
    \end{subfigure}%
    \newcommand{\printalgexample}[3]{%
        \begin{subfigure}[b]{.25\linewidth}
            \centering\includegraphics[width=.9\textwidth]{sel/#1.pdf}
            \caption{#3}\label{fig:algexample:#2}
        \end{subfigure}%
    }%
    \printalgexample{CNN}{cnn}{\CNN (253 pts)}%
    \printalgexample{FCNN}{fcnn}{\FCNN (185 pts)}%
    \printalgexample{SFCNN}{sfcnn}{\SFCNN (190 pts)}%
    
    \bigskip
    \printalgexample{MSS}{mss}{\MSS (234 pts)}%
    \printalgexample{RSS}{rss}{\RSS (192 pts)}%
    \printalgexample{VSS}{vss}{\VSS (193 pts)}%
    \printalgexample{NET}{net}{\NET (841 pts)}%
    
    \caption{An illustrative example of the subsets selected by different condensation algorithms from an initial training set \setP in $\RE^2$ of $10^4$ points. The list includes well-known algorithms like \CNN, \FCNN, \MSS, \RSS, \VSS, and \NET. We propose \SFCNN as a simple modification of \FCNN that can be successfully upper-bounded.}\label{fig:algexample}
\end{figure*}

Most algorithmic research for this problem has focused on heuristics; for comprehensive surveys, see \cite{DBLP:conf/jcdcg/Toussaint02,Toussaint02proximitygraphs,jankowski2004comparison}.
Out of the many heuristics proposed for this problem, \FCNN~\cite{angiulli2007fast} stands out due to its quadratic worst-case time complexity, and most importantly, its observed efficiency when applied to real-world training sets.
Alternatives include
    \CNN~\cite{Hart:2006:CNN:2263267.2267647},
    \MSS~\cite{barandela2005decision},
    \RSS~\cite{2019arXiv190412142F}, and
    \VSS~\cite{2019arXiv190412142F}.
These algorithms also run in quadratic time, except for \CNN, which has cubic runtime, and was the first algorithm proposed for condensation. See Figure~\ref{fig:algexample} for an illustrative comparison between these heuristics.

While such heuristics have been extensively studied experimentally~\cite{Garcia:2012:PSN:2122272.2122582}, theoretical results are scarce. Only recently in CCCG'19~\cite{2019arXiv190412142F}, we have shown that the size of the subset selected by \MSS cannot be bounded. On the other hand, we proved that the size of the subset selected by both \RSS and \VSS can be upper-bounded. However, until now, it remained open whether similar results could be achieved for \FCNN.

\paragraph{Contributions.}
In this paper, we settle the question of whether the size of the subsets selected by \FCNN can be upper-bounded. Our results are summarized as follows:
\begin{itemize}[leftmargin=*]
    \setlength\itemsep{-2pt}
    \item There exist training sets for which the subset selected by \FCNN is unbounded, particularly, when compared to the selection of other algorithms (\eg \RSS).
    \item We propose a modification of \FCNN, namely \SFCNN, for which we prove the following results: 
    \vspace*{-2pt}
    \begin{itemize}[leftmargin=*]
        \setlength\topsep{0pt}
        \setlength\itemsep{0pt}
        \item The size of the subset selected by \SFCNN has an upper-bound, similar to the one known for~\RSS.
        \item \SFCNN computes a tight approximation of the mini\-mum cardinality consistent subset of \setP.
    \end{itemize}
\end{itemize}

\paragraph{Preliminaries.}
Given any point $q \in \metricSet$ in the metric space, its nearest-neighbor, denoted by $\nn{q}$, is the closest point of \setP according the the distance function $\metricFunc$. Given a point $p \in \setP$, any other point of $\setP$ whose label differs from $p$'s is called an \emph{enemy} of $p$. The closest such point is called $p$'s \emph{nearest-enemy}, denoted by $\nenemy{p}$.

Clearly, the size of a condensed subset should depend on the spatial characteristics of the classes in the training set. For example, a consistent subset for two spatially well separated clusters should be smaller than the subset for one with two classes that have a high degree of overlap. To model this intrinsic complexity, define $\numNE$ to be the number of nearest-enemy points of \setP, \ie the cardinality of set $\{\nenemy{p} \mid p \in \setP\}$.

This has been previously used~\cite{2019arXiv190412142F} to prove useful upper-bounds for \RSS, and to show negative results for \MSS. In particular, it has been shown that \RSS selects $\bigOh(\kappa\,(3/\pi)^{d-1})$ points in $d$-dimensional Euclidean space, while \MSS's selection cannot be \mbox{bounded in terms of $\kappa$.}

\section{Nearby Points are Problematic}

\begin{algorithm}
 \DontPrintSemicolon
 \vspace*{0.1cm}
 \KwIn{Initial training set \setP}
 \KwOut{Consistent subset $\setR \subseteq \setP$}
 $\setR \gets \phi$\;
 $S \gets \textup{centroids}(\setP)$\;
 \While{$S \neq \phi$}{
  $\setR \gets \setR \cup S$\;
  $S \gets \phi$\;
  \ForEach{$p \in \setR$}{
   $S \gets S \cup \{ \textup{rep}(p,\textup{voren}(p,\setR,\setP)) \}$\;
  }
 }
 \KwRet{\setR}
 \vspace*{0.1cm}
 \caption{\FCNN}
 \label{alg:fcnn}
\end{algorithm}

Now consider the \FCNN algorithm~\cite{angiulli2007fast}. It follows an iterative incremental approach to build a consistent subset of \setP (see Algorithm~\ref{alg:fcnn} for a formal description). While not immediately evident, \FCNN it runs in $\bigOh(nm)$ worst-case time, where $m$ is the final size of the selected~subset.

The algorithm begins by selecting one point per class, in particular, the centroid of each class%
\footnote{For each class, its centroid is defined as the closest point of \setP to the geometrical center of all points of this class.}.
Then, it begins the iterative process, selecting other points until the subset is consistent. During each iteration, the algorithm identifies all points of \setP that are misclassified with respect to the current subset, and adds some of these points to the subset. Formally, for every point $p$ already in the subset, \FCNN selects one \emph{representative} among non-selected points, whose nearest-neighbor is $p$ and that belong to a different class than $p$. That is, the representative is selected from the set $\textup{voren}(p,\setR,\setP) = \{ q \in \setP \mid \nn{q,\setR} = p \wedge l(q) \neq l(p) \}$. Usually, the representative chosen is the one closest to $p$, although different approaches can be applied.

However, there is an issue with this algorithm. During any given iteration, nothing prevents the representatives of two neighboring points in \FCNN to be arbitrarily close to each other. This observation can be exploited to obtain the following result:

\begin{theorem}
\label{thm:fcnn-unbounded}
There exists a training set $\setP \subset \RE^d$ in Euclidean space, with constant number of classes, for which \FCNN selects $\Omega(\numNE/\xi)$ points, for any $0 < \xi < 1$.
\end{theorem}

The remaining of this section addresses the proof of this theorem, by carefully constructing a training set \setP in $\RE^3$ that exhibits the undesirable behavior in the selection process of the \FCNN algorithm.

Without loss of generality, let $\xi = 1/2^{t}$ for some value $t>3$, we construct a training set $\setP \subset \RE^3$ with a constant number of classes and the number of nearest-enemy points $\numNE$ equal to $\bigOh(1/\xi)$, for which \FCNN is forced to select $\bigOh(1/\xi^2)$ points.
As mentioned above, the key downside of the algorithm occurs when points are added to the subset in the same iteration, as they can be arbitrarily close to each other. We exploit this behavior to force the algorithm to select $\bigOh(1/\xi)$ such points on each iteration.

Intuitively, the training set \setP consists of several layers of points arranged parallel to the $xy$-plane, and stacked on top of each other around the $z$-axis (see Figure~\ref{fig:fcnn:kappa}). Each layer is a disk-like arrangement, formed by a center point and points at distance 1 from this center. Define the \emph{backbone points} of \setP to be the center points $c_i = 2i \vec{v}_z$ for $i \geq 0$. We now describe the different arrangements of points as follows (see Figure~\ref{fig:fcnn:kappa}):
\begin{align*}
\mathcal{B} &=
    c_0 \cup \left\{ y_j = c_0 + \vec{v}_x R_z(j \pi/4) \right\}^{8}_{j=1}\\[0.5ex]
\mathcal{M}_i &=
    \left\{c_{2i}, c_{2i+1}, m_i = (c_{2i}+c_{2i+1})/2 \right\}\\
    &\cup
    \left\{ r_{ij} = c_{2i} + \vec{v}_x R_z \left(j \pi/2^{1+i} \right) \right\}^{2^{2+i}}_{j=1}\\
    &\cup
    \left\{ b_{ij} = c_{2i+1} + \vec{v}_x R_z \left(j \pi/2^{1+i} \right) \right\}^{2^{2+i}}_{j=1}\\
    &\cup
    \left\{ w_{ij}\!=\!c_{2i+1}\!+\!\vec{v}_x R_z \left((j\!+\!1/2) \pi/2^{1+i}\!-\!\xi^2 \right) \right\}^{2^{2+i}}_{j=1}\\[0.5ex]
\mathcal{R}_i &=
    \left\{c_{2i}, c_{2i+1} \right\}\\
    &\cup
    \left\{ r_{ij} = c_{2i} + \vec{v}_x R_z \left(2j\pi\xi \right) \right\}^{1/\xi}_{j=1}\\
    &\cup
    \left\{ b_{ij} = c_{2i+1} + \vec{v}_x R_z \left(2j\pi\xi \right) \right\}^{1/\xi}_{j=1}
\end{align*}

\begin{figure*}[t!]
    \begin{minipage}{0.48\linewidth}
        \begin{subfigure}[b]{\linewidth}
            \centering\includegraphics[width=0.9\textwidth]{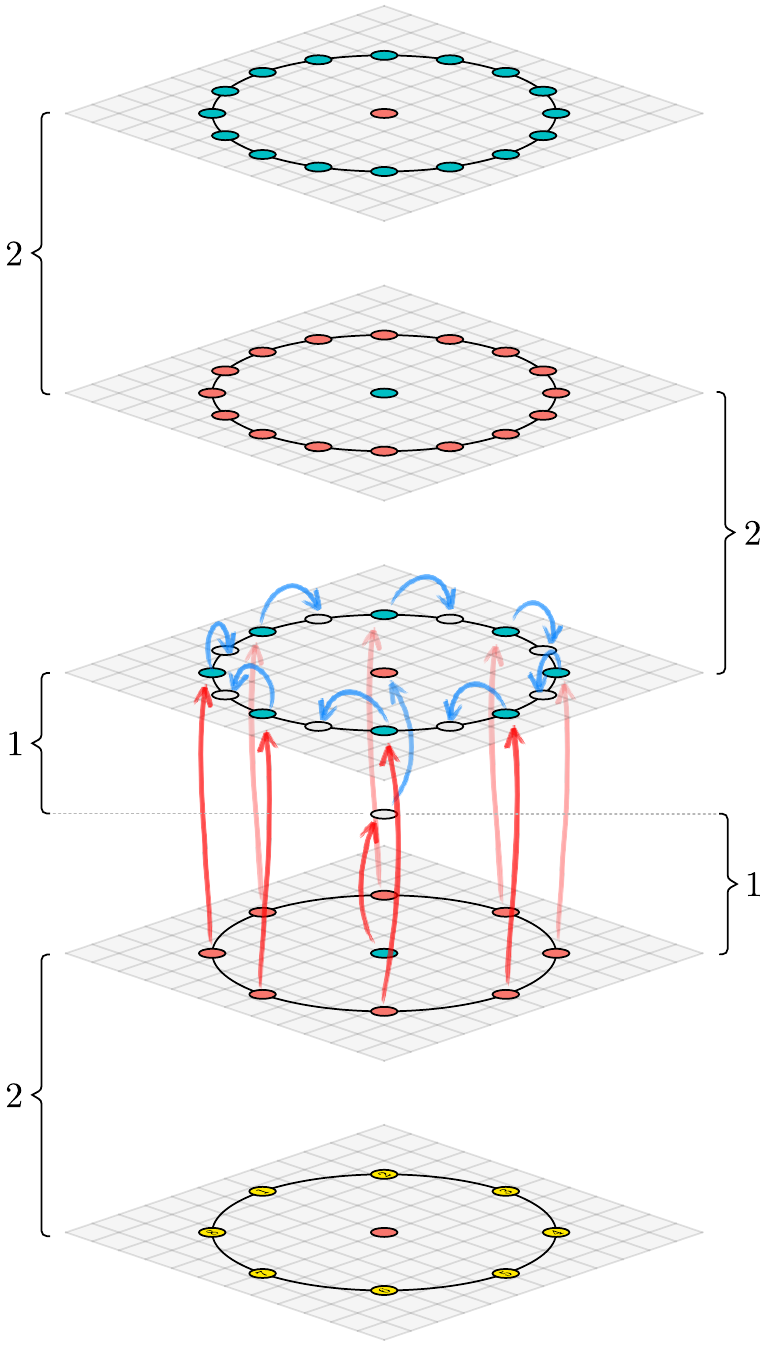}
            \caption{Entire arrangement of points, by stacking the diffe\-rent arrangements along the $z$-axis. The arrows illustrate the selection process of \FCNN on a multiplicative arrangement~$\mathcal{M}_i$.}\label{fig:fcnn:kappa:all}
        \end{subfigure}%
    \end{minipage}
    \hfill
    \begin{minipage}{0.48\linewidth}
        \begin{subfigure}[b]{\linewidth}
            \centering\includegraphics[width=0.9\textwidth]{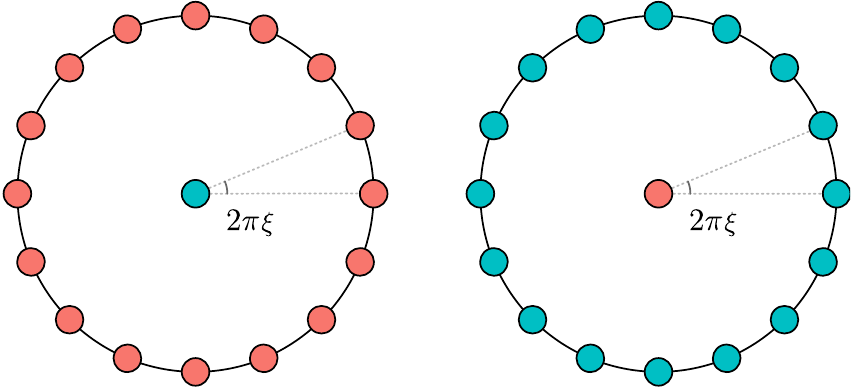}
            \caption{A repetitive arrangement $\mathcal{R}_i$. This is used to maintain the number of selected representatives to be $\bigOh(1/\xi)$ during each remaining iteration of the algorithm.}\label{fig:fcnn:kappa:rep}
        \end{subfigure}\\[2ex]
        \bigskip
        \begin{subfigure}[b]{\linewidth}
            \centering\includegraphics[width=0.9\textwidth]{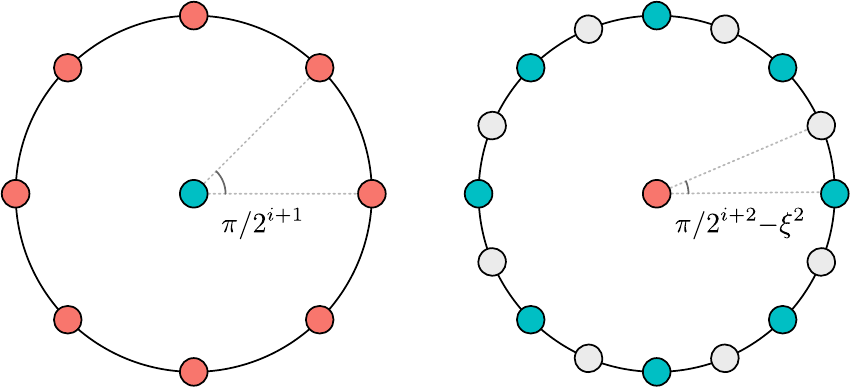}
            \caption{A multiplier arrangement $\mathcal{M}_i$. This forces \FCNN to double the number of selected representatives around the circumference after two iterations.}\label{fig:fcnn:kappa:mul}
        \end{subfigure}\\
        \bigskip
        \begin{subfigure}[b]{\linewidth}
            \centering\includegraphics[width=0.4\textwidth]{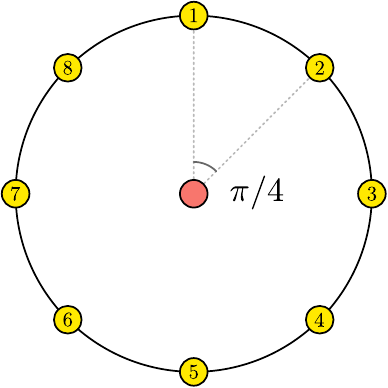}
            \caption{Base arrangement $\mathcal{B}$. Each point in the circumference belongs to one unique class $\{\textup{1},\dots,\textup{8}\}$, here colored in yellow and numbered for clarity.}\label{fig:fcnn:kappa:base}
        \end{subfigure}%
    \end{minipage}
    \vspace*{-10pt}
    \caption{Example of a training set $\setP \subset \RE^3$ for which \FCNN selects $\Omega(\numNE/\xi)$ points.}\label{fig:fcnn:kappa}
\end{figure*}

These points belong to one of 11 classes, defined by the set $\{ 1, \dots, 8, \textup{red}, \textup{blue}, \textup{white} \}$. Then, we define the labeling function $l$ as follows: $l(c_i)$ is red when $i$ is even and blue when $i$ is odd, $l(m_i)$ is white, $l(y_j)$ is the $j$-th class, $l(r_{ij})$ is red, $l(b_{ij})$ is blue, and $l(w_{ij})$ is white.

\paragraph{Base arrangement ($\mathcal{B}$).}
Consists of one single layer of points, with one red center point $c_0$ and 8 points $y_j$ in the circumference of the unit disk (parallel to the $xy$-plane), each labeled with a unique class $j$ (see Figure~\ref{fig:fcnn:kappa:base}).

The goal of this arrangement is that each of these points is the centroid of its corresponding class. The centroids of the blue and white classes can be fixed to be far enough, so we won't consider them for now. Hence, the first iteration of \FCNN will add all the points of $\mathcal{B}$. In the next iteration, each of these points will select a representative in the arrangement above. Clearly, the size of $\mathcal{B}$ is 9, and it contributes with 8 nearest-enemy points in total.
    
\paragraph{Multiplier arrangement ($\mathcal{M}_i$).}
Our final goal is to have $\bigOh(1/\xi)$ arbitrarily close points selecting representatives on a single iteration;. Initially, we only have 9, the ones in the base arrangement. While this could be simply achieved with $\bigOh(1/\xi)$ points in $\mathcal{B}$ each with a unique class, we want to use a constant number of classes. Instead, we use each multiplier arrangement to double the number of representatives selected.

$\mathcal{M}_i$ consists of (1) a layer with a blue center $c_{2i}$ and $2^{2+i}$ red points $r_{ij}$ around the unit disk's circumference, (2) a layer with a red center $c_{2i+1}$ and $2^{3+i}$ blue $b_{ij}$ and white $w_{ij}$ points around the unit disk's circumference, and (3) a white center point $m_i$ in the middle between the red and blue center points (see Figure~\ref{fig:fcnn:kappa:mul}).
Suppose at iteration $3i-1$ all the points $r_{ij}$ and $c_{2i}$ of the first layer are added as representatives of the previous arrangement, which is given for $\mathcal{M}_1$ from the selection of $\mathcal{B}$.
Then, during iteration $3i$ each $r_{ij}$ adds the point $b_{ij}$ right above, while $c_{2i}$ adds point $m_i$ (see the red arrows in Figure~\ref{fig:fcnn:kappa:all}). Finally, during iteration $3i+1$, $m_i$ adds $c_{2i+1}$, and each $b_{ij}$ adds point $w_{ij}$ as its the closest point inside the voronoi cell of $b_{ij}$ (see the blue arrows in Figure~\ref{fig:fcnn:kappa:all}). Now, with all the points of this layer added, each continues to select points in the following arrangement (either $\mathcal{M}_{i+1}$ or $\mathcal{R}_{i+1}$).

The size of each $\mathcal{M}_i$ is $3(1+2^{2+i}) = \bigOh(2^{3+i})$, and contributes with $3+2(2^{2+i}) = \bigOh(2^{3+i})$ to the total number of nearest-enemy points.
Thus, we stack $\mathcal{M}_i$ arrangements for $i \in [1,\dots,t-3]$, such that the last of these selects $1/\xi = 2^t$ points.

\paragraph{Repetitive arrangement ($\mathcal{R}_i$).}
Once the algorithm reaches the last multiplier layer $\mathcal{M}_{t-3}$, it will select $1/\xi$ points during the following iteration. The repeti\-tive arrangement allows us to continue selecting these many points on every iteration, while only increasing the number of nearest-enemy points by a constant. This arrangement consists of (1) a first layer with a blue center $c_{2i}$ surrounded by $1/\xi$ red points $r_{ij}$ around the unit disk circumference, and (2) a second layer with red center $c_{2i+1}$ and blue points $b_{ij}$ in the circumference (see Figure~\ref{fig:fcnn:kappa:rep}). Once the first layer is added all in a single iteration, during the following iteration $c_{2i}$ adds $c_{2i+1}$, and each $r_{ij}$ adds $b_{ij}$.

The size of each $\mathcal{R}_i$ is $2(1+1/\xi) = \bigOh(1/\xi)$, and it contributes with 4 points to the total number of nearest-enemy points.
Now, we stack $\bigOh(1/\xi)$ such arrangements $\mathcal{R}_i$ for $i \in [t-2, \dots, 1/\xi]$, such that we obtain the desired ratio between selected points and number of nearest-enemy points of the training set.

\paragraph{The training set.}
After defining all the necessary point arrangements, and recalling that $t = \log{1/\xi}$,
we put these arrangements together to define the training set \setP as follows:
\[
\setP = \mathcal{B} \,\bigcup\limits_{i=1}^{t-3} \mathcal{M}_i \,\bigcup\limits_{i=t-2}^{1/\xi} \mathcal{R}_i \cup \mathcal{F}
\]
where $\mathcal{F}$ is an additional set of points to fix the centroids of \setP. These extra points are located far enough from the remaining points of \setP, and are carefully placed such that the centroids of \setP are all the points of $\mathcal{B}$, plus a blue and white point from $\mathcal{F}$. Additionally, all the points of $\mathcal{F}$ should be closer to its corresponding class centroid than to any enemy centroid, and they should increase the number of nearest-enemy points by a constant. This can be done with a bounded number of extra points.

All together, by adding up the corresponding terms, the ratio between the size of \FCNN and $\numNE$ (the number of nearest-enemy points of \setP) is $\bigOh(1/\xi)$. Therefore, on this training set, \FCNN selects $\bigOh(\numNE/\xi)$ points.

\section{Keeping Distance: One by One}

Evidently, adding points in batch on every iteration of the algorithm prevents \FCNN to have provable guarantees on the size of its selected subset, just as \RSS provides.
However, this design choice is not key for any of the features of the algorithm.

Therefore, we propose to modify \FCNN such that only \emph{one} single representative is added to the subset on each iteration. We call this new algorithm \SFCNN or \emph{Single} \FCNN. Basically, the only difference between the original \FCNN and \SFCNN is on line 4 of Algorithm~\ref{alg:fcnn}, where \setR is updated by selecting one single point from the set of representatives $S$, as follows:
\[
    \setR \gets \setR \cup \{ \text{Choose one point of } S\}
\]

While extremely simple, this change in the selection process of \SFCNN allows us to successfully analyze the size of its selected subset in terms of $\numNE$, and even prove that it approximates the consistent subset of minimum cardinality.

\paragraph{Size Upper-Bound.}
To this end, we first need to introduce some termino\-logy.
Through a suitable uniform scaling, we may assume that the \emph{diameter} of \setP (that is, the maximum distance between any two points in the training set) is 1. The \emph{spread} of \setP, denoted as $\spread$, is the ratio between the largest and smallest distances in \setP. Define the \emph{margin} of \setP, denoted $\gamma$, to be the smallest nearest-enemy distance in \setP. Clearly, $1/\gamma \leq \spread$.

The metric space \metricSpace is said to be \emph{doubling}~\cite{heinonen2012lectures} if there exist some bounded value $\lambda$ such~that any metric ball of radius $r$ can be covered with at most $\lambda$ metric balls of radius $r/2$. Its \emph{doubling dimension} is the base-2 logarithm of $\lambda$, denoted as $\ddim = \log{\lambda}$. Throughout, we assume that $\ddim$ is a constant, which means that multiplicative factors depending on $\ddim$ may be hidden in our asymptotic notation. Many natural metric spaces of interest are doubling, including $d$-dimensional Euclidean space whose doubling dimension is $\Theta(d)$. It is well know that for any subset $R \subseteq \metricSet$ with some spread $\spread_R$, the size of $R$ is bounded by $|R| \leq \lceil\spread_R\rceil^{\ddim+1}$.

\begin{theorem}
\label{thm:sfcnn-size}
\SFCNN selects a subset of size:
\[
\bigOh\!\left( \numNE \log{\frac{1}{\gamma}}\ 4^{\ddim+1}\right)
\] 
\end{theorem}

\begin{proof}
This follows by a charging argument on each nearest-enemy point in the training set.
Consider one such point $p \in \{ \nenemy{r} \mid r \in \setP\}$ and a value $\sigma \in [\gamma,1]$. We define $\setR_{p,\sigma}$ to be the subset of points from \SFCNN whose nearest-enemy is $p$, and whose distance to $p$ is between $\sigma$ and $2\sigma$. That is, $\setR_{p,\sigma} = \{ r \in \setR \mid \nenemy{r}=p \wedge \dist{r}{p} \in [\sigma,2\sigma) \}$. These subsets define a partitioning of \setR when considering all nearest-enemy points of \setP, and values of $\sigma = \gamma\, 2^i$ for $i=\{0,1,2,\dots,\lceil\log{\frac{1}{\gamma}}\rceil\}$.

Consider any two points $a, b \in \setR_{p,\sigma}$ in these subsets. Assume \wwlog that point $a$ was selected by the algorithm before point $b$ (\ie in a prior iteration). We show that $\dist{a}{b} \geq \sigma$. By contradiction, assume that $\dist{a}{b} < \sigma$, which immediately implies that $a$ and $b$ belong to the same class. Moreover, recalling that $b$'s nearest-enemy in \setP is $p$, at distance $\dist{b}{p} \geq \sigma$, this implies that $b$ is closer to $a$ than to any enemy in \setR. Therefore, by the definition of the \emph{voren} function, $b$ could never be selected by \SFCNN, which is a contradiction.

This proves that $\dist{a}{b} \geq \sigma$.
Additionally, we know that $\dist{a}{b} \leq \dist{a}{p} + \dist{p}{b} \leq 4\sigma$. Thus, using a simple packing argument with the known properties of doubling spaces, we have that $|\setR'_{p,\sigma}| \leq 4^{\ddim+1}$.

Altogether, by counting over all the $\setR_{p,\sigma}$ sets for every nearest-enemy in the training set and values of $\sigma$, the size of \setR is upper-bounded by $|\setR| \leq \numNE \left\lceil\log{1/\gamma} \right\rceil 4^{\ddim+1}$. This completes the proof.
\end{proof}

\paragraph{An Approximation Algorithm.}
Denote \textsc{Min-CS} as the problem of computing a minimum cardinality consistent subset of \setP. This problem is known to be NP-hard~\cite{Wilfong:1991:NNP:109648.109673,Zukhba:2010:NPP:1921730.1921735,khodamoradi2018consistent}, even to approximate~\cite{gottlieb2014near} in polynomial time within a factor of $2^{({\ddim \log{(1/\gamma)})}^{1-o(1)}}$.

As previously mentioned, the \NET algorithm~\cite{gottlieb2014near} computes a tight approximation for the \textsc{Min-CS} problem. The algorithm is rather simple: it just computes a $\gamma$-net of \setP, where $\gamma$ is the margin (the smallest nearest-enemy distance in \setP). This clearly results in a consistent subset of \setP, whose size is at most $\left\lceil 1/\gamma \right\rceil^{\ddim+1}$. A similar result can be proven for \SFCNN.

\begin{theorem}
\label{thm:sfcnn-approx-factor}
\SFCNN computes a tight approximation for the \textsc{Min-CS} problem.
\end{theorem}

\begin{proof}
This follows from a direct comparison to the resulting subset of the \NET algorithm. For any point $p \in \NET$, let $B_p$ be the set of points of \setP ``covered'' by $p$, that is, whose distance to $p$ is at most $\gamma$. By the covering property of nets, this defines a partition on \setP when considering every point $p$ selected by \NET.

Let's analyze the size of $B_p \cap \setR$, that is, for any given $B_p$ how many points could have been selected by the \SFCNN algorithm. Let $a, b \in B_p \cap \setR$ be two such points, where without loss of generality, point $a$ was selected in an iteration before $b$. Both $a$ and $b$ must belong to the same class as $p$, as their distance to $p$ is at most $\gamma$, which is the smallest nearest-enemy distance in \setP. Moreover, by the definition of the \emph{voren} function, it is easy to show that $\dist{a}{b} \geq \gamma$. By a simple packing argument in doubling metrics, the size of any $B_p \cap \setR$ is at most $2^{\ddim+1}$. All together, we have that the size of the subset selected by \SFCNN is $2^{\ddim+1}\,|\NET| = \bigOh\left( (1/\gamma)^{\ddim+1} \right)$.
\end{proof}

\section{Experimental Results}

The importance of \FCNN relies on its performance in practice, despite the lack of theoretical guarantees. A natural question is whether the simple change we proposed on the algorithm, negatively affects its performance in real-world training sets.

Thus, to get a clearer impression of the relevance of these results in practice, we performed experimental trials on several training sets, both synthetically generated and widely used benchmarks. First, we consider 21 training sets from the UCI \emph{Machine Learning Repository}\footnote{\url{https://archive.ics.uci.edu/ml/index.php}} which are commonly used in the literature to evaluate condensation algorithms~\cite{Garcia:2012:PSN:2122272.2122582}. These consist of a number of points ranging from 150 to $58000$, in $d$-dimensional Euclidean space with $d$ between 2 and 64, and 2 to 26 classes.
We also generated some synthetic training sets, containing $10^5$ uniformly distributed points, in 2 to 3 dimensions, and 3 classes. 
All training sets used in these experimental trials are summarized in Table~\ref{table:data}. The implementation of the algorithms, training sets used, and raw results, are publicly available\footnote{\url{https://github.com/afloresv/nnc/}}.

\begin{table}[h!]
\footnotesize
\centering
\begin{tabular}{||c|cccc||}
\hline\rule{0pt}{9pt}
Training set & $n$ & $d$ & $c$ & $\numNE\ (\%)$ \\ \hline
banana & 5300 & 2 & 2 & 811 (15.30\%) \\
cleveland & 297 & 13 & 5 & 125 (42.09\%) \\
glass & 214 & 9 & 6 & 87 (40.65\%) \\
iris & 150 & 4 & 3 & 20 (13.33\%) \\
iris2d & 150 & 2 & 3 & 13 (8.67\%) \\
letter & 20000 & 16 & 26 & 6100 (30.50\%) \\
magic & 19020 & 10 & 2 & 5191 (27.29\%) \\
monk & 432 & 6 & 2 & 300 (69.44\%) \\
optdigits & 5620 & 64 & 10 & 1245 (22.15\%) \\
pageblocks & 5472 & 10 & 5 & 429 (7.84\%) \\
penbased & 10992 & 16 & 10 & 1352 (12.30\%) \\
pima & 768 & 8 & 2 & 293 (38.15\%) \\
ring & 7400 & 20 & 2 & 2369 (32.01\%) \\
satimage & 6435 & 36 & 6 & 1167 (18.14\%) \\
segmentation & 2100 & 19 & 7 & 398 (18.95\%) \\
shuttle & 58000 & 9 & 7 & 920 (1.59\%) \\
thyroid & 7200 & 21 & 3 & 779 (10.82\%) \\
twonorm & 7400 & 20 & 2 & 1298 (17.54\%) \\
wdbc & 569 & 30 & 2 & 123 (21.62\%) \\
wine & 178 & 13 & 3 & 37 (20.79\%) \\
wisconsin & 683 & 9 & 2 & 35 (5.12\%) \\
v-100000-2-3-15 & 100000 & 2 & 3 & 1909 (1.90\%) \\
v-100000-2-3-5 & 100000 & 2 & 3 & 788 (0.78\%) \\
v-100000-3-3-15 & 100000 & 3 & 3 & 7043 (7.04\%) \\
v-100000-3-3-5 & 100000 & 3 & 3 & 3738 (3.73\%) \\
v-100000-4-3-15 & 100000 & 4 & 3 & 13027 (13.02\%) \\
v-100000-4-3-5 & 100000 & 4 & 3 & 10826 (10.82\%) \\
v-100000-5-3-15 & 100000 & 5 & 3 & 22255 (22.25\%) \\
v-100000-5-3-5 & 100000 & 5 & 3 & 17705 (17.70\%) \\
\hline
\end{tabular}
\vspace*{.3cm}
\caption{Training sets used to evaluate the performance of condensation algorithms. Indicates the number of points $n$, dimensions $d$ (Euclidean space), classes $c$, nearest-enemy points $\numNE$ (also in percentage \emph{w.r.t.} $n$).}\label{table:data}
\end{table}

We test 7 different condensation algorithms, namely \FCNN, \SFCNN, \RSS, \VSS, \MSS, \CNN and \NET. To compare their results, we consider their runtime and the size of the selected subset. Clearly, these values might differ greatly on training sets whose size are too distinct. Therefore, before comparing the raw results, these are normalized. The runtime of an algorithm for a given training set is normalized by dividing it by $n$, the size of the training set. The size of the selected subset is normalized by dividing it by $\numNE$, the number of nearest-enemy points in the training set, which characterizes the complexity of the boundaries between classes.

\begin{figure}[h!]
    \centering
    \includegraphics[width=\linewidth]{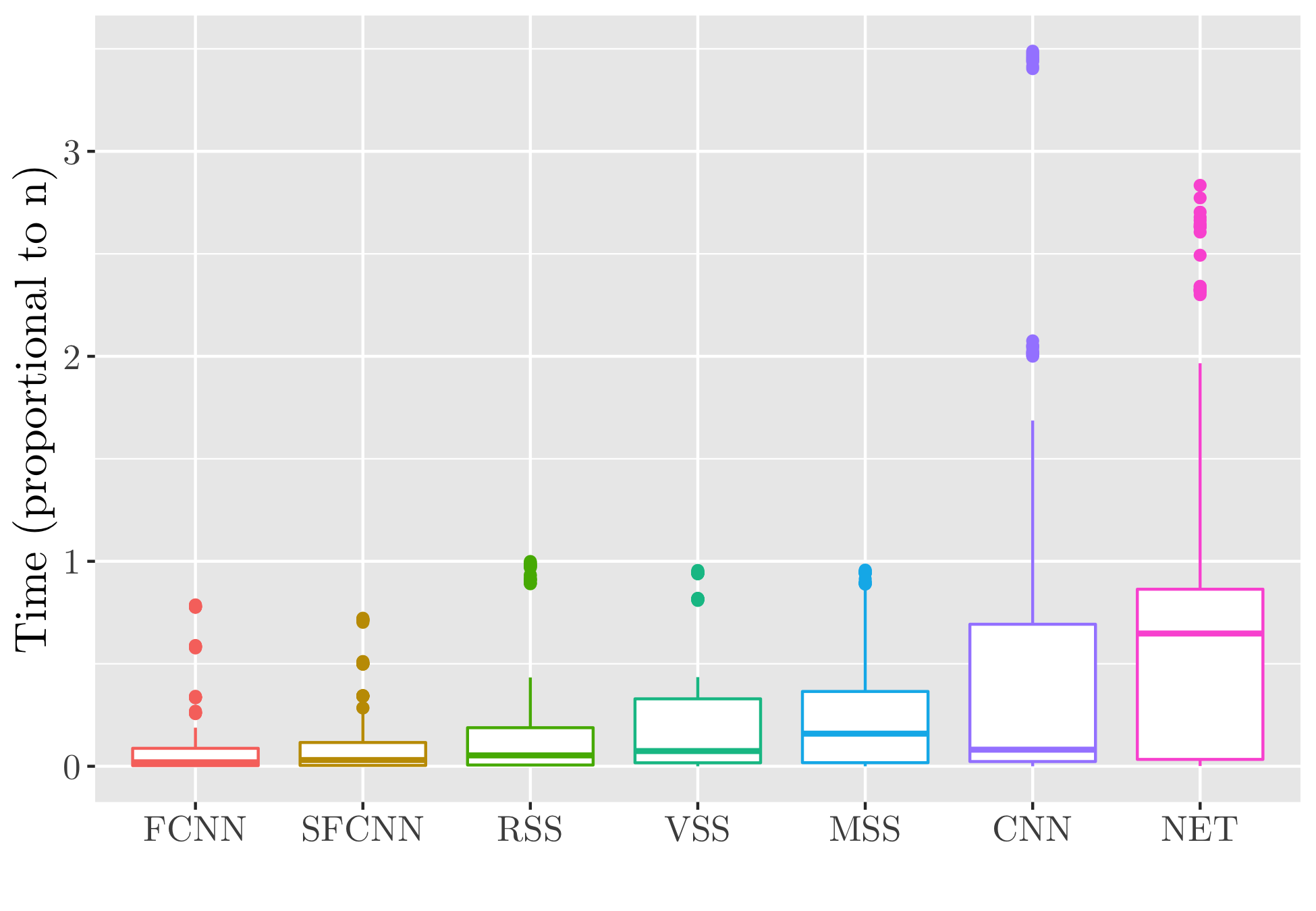}
    \vspace*{-25pt}
    \caption{Running time.}\label{fig:exp:standard:time}
\end{figure}

Figures \ref{fig:exp:standard:time} and \ref{fig:exp:standard:size} summarize the experimental results. Evidently, the performance of \SFCNN is equivalent to the original \FCNN algorithm, both in terms of runtime and the size of their selected subsets, showing that the proposed modification does not affect the behavior of the algorithm in real-world training sets. Both \FCNN and \SFCNN outperform other condensation algorithms in terms of runtime, while their subset size is comparable in all cases, with the exception of the \NET algorithm.

\begin{figure}[h!]
    \centering
    \includegraphics[width=\linewidth]{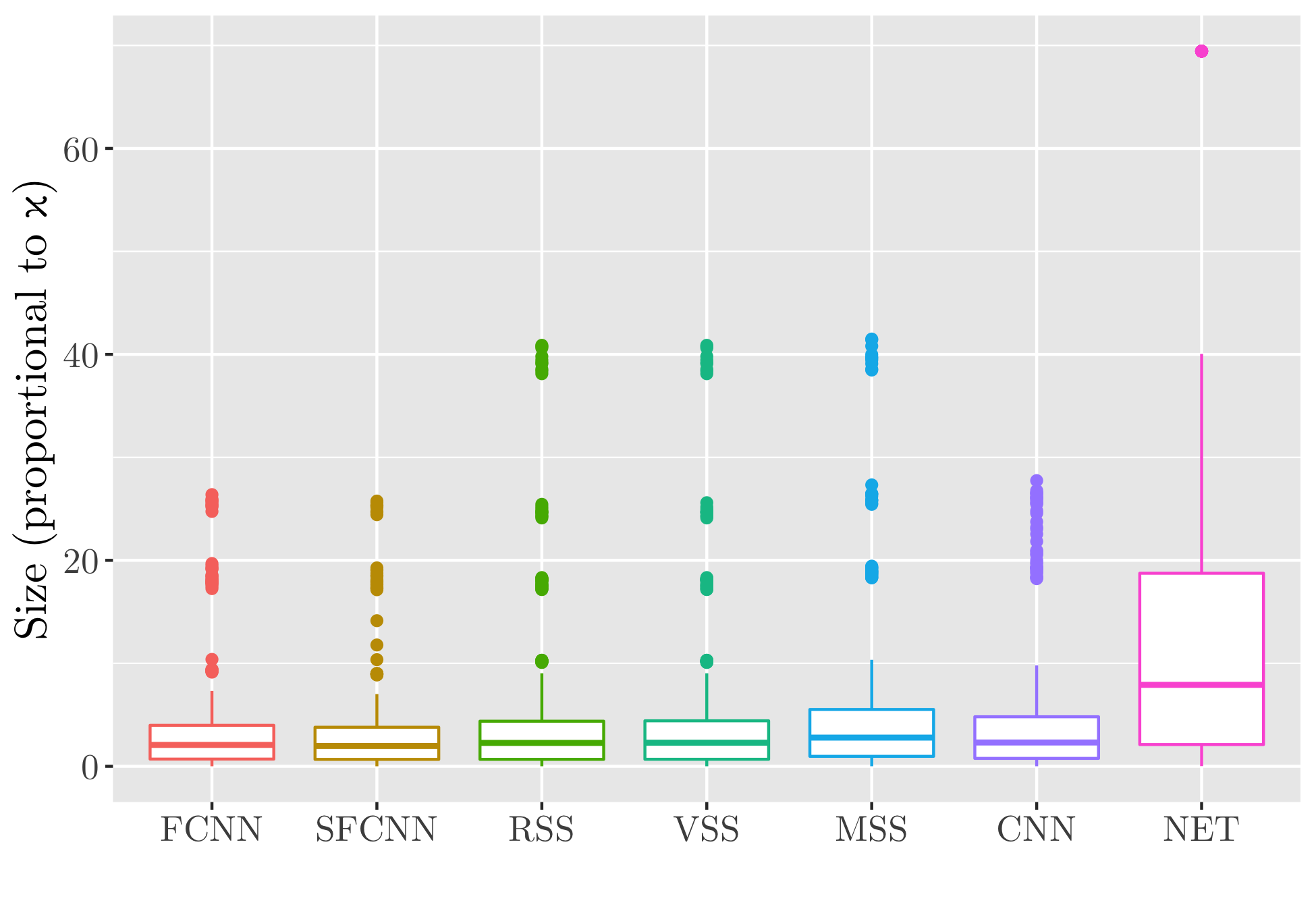}
    \vspace*{-25pt}
    \caption{Size of the selected subsets.}\label{fig:exp:standard:size}
\end{figure}



\small
\bibliographystyle{abbrv}
\bibliography{nnc}

\end{document}